\documentclass{article}
\pdfoutput=1
\usepackage{amsthm}
\usepackage{graphicx}
\usepackage{lmodern}

\author{
Jean-Guillaume Dumas\footnote{Universit\'e de Grenoble Alpes. 
Laboratoire LJK, umr CNRS. 
51, av. des Math\'ematiques, F38041 Grenoble,
France.
\href{mailto:Jean-Guillaume.Dumas@imag.fr}{Jean-Guillaume.Dumas@imag.fr},
}
\and
Cl\'ement Pernet\footnote{Universit\'e de Grenoble Alpes.
Laboratoire LIP,
Inria, CNRS, UCBL, ENS de Lyon, 
46, All\'ee d'Italie, F69364 Lyon Cedex 07 France.
\href{mailto:Clement.Pernet@imag.fr}{Clement.Pernet@imag.fr},
}
\and 
Ziad Sultan\footnote{
Universit\'e de Grenoble Alpes.
Laboratoires LJK and LIG,
Inria, CNRS.
Inovall\'ee, 655, av. de l'Europe, F38334 St Ismier Cedex, France.
\href{mailto:Ziad.Sultan@imag.fr}{Ziad.Sultan@imag.fr}.
}
}

\newcommand{\strechparskip}[1]{}
\newcommand{\strechparsep}[1]{}
\newcommand{\customvspace}[1]{}
\newcommand{\category}[3]{}
\newcommand{\terms}[1]{}
\newcommand{\keywords}[1]{}

\newcommand{\makeconference}{}

\newenvironment{smatrix}{\begin{bmatrix}}{\end{bmatrix}}

\newfont{\seaddfnt}{phvr8t at 8pt}

\usepackage{mdwlist}
\usepackage{multicol}
\usepackage{enumerate}
\usepackage{paralist}
\usepackage[utf8]{inputenc}
\usepackage{xspace}
\usepackage{amsmath,amssymb}
\usepackage{booktabs}
\graphicspath{{Pictures/}}  
\usepackage{algorithm}
\usepackage{algpseudocode}
\algrenewcommand\algorithmicrequire{\textbf{Input:}}
\algrenewcommand\algorithmicensure{\textbf{Output:}}
\algrenewcommand\algorithmicreturn{\textbf{Return}}
\algrenewcommand\Return{\State\algorithmicreturn{} }%

\usepackage{multirow}
\newcommand{\fflasffpack}{\texttt{FFLAS-FFPACK}\xspace}
\newcommand{\RPM}[1]{\ensuremath{\mathcal{R}_{#1}\xspace}} %
\newcommand{\RS}[1]{\ensuremath{\text{RowSupp}({#1})\xspace}}  %
\newcommand{\CS}[1]{\ensuremath{\text{ColSupp}({#1})\xspace}}  %
\newcommand{\K}{\mathrm{K}\xspace}

\newcommand{\RRP}{\text{RowRP}}
\newcommand{\CRP}{\text{ColRP}}
\newcommand{\LA}{\ensuremath{\overline{A}_1\xspace}}

\newcommand{\lup}{\text{LUP}\xspace}

\newcommand{\cupd}{\text{CUP}\xspace}
\newcommand{\plu}{\text{PLU}\xspace}
\newcommand{\pluq}{\text{PLUQ}\xspace}

\newcommand{\rank}{\text{rank}\xspace}

\newtheorem{theorem}{Theorem}

\newtheorem{example}{Example}
\newtheorem{lemma}{Lemma}
\newtheorem{definition}{Definition}
\newtheorem{corollary}{Corollary}
\newtheorem{remark}{Remark}

\title{Computing the Rank Profile Matrix\thanks{\small This work is partly funded by the HPAC project of the French Agence Nationale de la Recherche (ANR~11~BS02~013).}}

\makeatletter
\usepackage{color,svgcolor}
\usepackage[plainpages=true]{hyperref}
\hypersetup{
pdftitle={Computing the Rank profile matrix},
pdfauthor={Jean-Guillaume Dumas, Clément Pernet and Ziad Sultan},
breaklinks=true, %
colorlinks=true,
 linkcolor=darkred,
 citecolor=blue,
 urlcolor=darkgreen,
}
\makeatother

\begin{document}
\setlength{\arraycolsep}{.3\arraycolsep}

\makeconference
\maketitle

\begin{abstract}
The row (resp. column) rank profile of a matrix describes the
stair-case shape of its row (resp. column) echelon form.	
In an ISSAC'13 paper, we proposed a recursive Gaussian elimination
that can compute simultaneously the row and column rank
profiles of a matrix as well as those of all of its leading
sub-matrices, in the same time as state of the art Gaussian
elimination algorithms.
Here we first study the conditions making a Gaussian
elimination algorithm reveal this information.
Therefore, we propose the definition of a new matrix invariant, 
the rank profile matrix, summarizing all information on the row and
column rank profiles of all the leading sub-matrices.
We also explore the conditions for a Gaussian
elimination algorithm to compute all or part of this invariant,
through the corresponding PLUQ decomposition.
As a consequence, we show that the classical iterative CUP decomposition
algorithm can actually be adapted to compute the rank profile matrix. 
Used, in a Crout variant, as a base-case to our ISSAC'13 implementation, it delivers a
significant improvement in efficiency.
Second, the row (resp. column) echelon form of a matrix are usually
computed via different dedicated triangular decompositions.
We show here that, from some PLUQ decompositions, it is possible
to recover the row and column echelon forms of a matrix and of
any of its leading sub-matrices thanks to an elementary post-processing algorithm.
\end{abstract}


\section{Introduction}

Triangular matrix decompositions are widely used in computational linear
algebra. Besides solving linear systems of equations, they are also
used to compute other objects more specific to exact arithmetic:
computing the rank,
sampling a vector from the null-space, computing echelon forms and
rank profiles.

The {\em row rank profile} (resp.\ {\em column rank profile}) of an $m\times n$ matrix $A$ with rank~$r$, denoted
by \RRP(A) (resp.\ \CRP(A)), is the
lexicographically smallest sequence of $r$ indices of linearly
independent rows (resp.\ columns) of $A$.
An $m\times n$ matrix has generic row (resp.\ column) rank profile if its row
(resp.\ column) rank profile is  $(1,..,r)$.
Lastly, an $m\times n$ matrix has generic rank profile if its $r$ first leading
principal minors are non-zero. Note that if a matrix has generic rank profile,
then its row and column rank profiles are generic, but the converse is false: the
matrix $\begin{smatrix}  0&1\\1&0\end{smatrix}$ does not have generic rank profile even if its row and column rank profiles
are generic.
The row support (resp.\ column support) of a matrix $A$, denoted by \RS{A}
(resp.\ \CS{A}), is the subset of indices of its non-zero rows (resp.\ columns). 

We recall that the row echelon form of an $m\times n$ matrix $A$ is an
upper triangular matrix $E=TA$, for a non-singular matrix $T$,  with the zero rows
of $E$ at the bottom and the non-zero rows in stair-case shape:
$\min\{j:a_{i,j}\neq0\} < \min\{j:a_{i+1,j}\neq0\}$.
As $T$ is non singular, the column rank profile of $A$ is that of $E$, and
therefore corresponds to the column indices of the leading elements in the
staircase. Similarly the row rank profile of $A$ is composed of the row indices of the
leading elements in the staircase of the column echelon form of $A$.

\paragraph{Rank profile and triangular matrix decompositions}
The rank profiles of a matrix and the triangular matrix decomposition obtained
by Gaussian elimination are strongly related.
The elimination of matrices with arbitrary rank profiles gives rise
to several matrix factorizations and many algorithmic variants.
In numerical linear algebra one often uses the PLUQ decomposition, with
$P$ and $Q$ permutation matrices, $L$ a lower unit triangular matrix
and $U$ an upper triangular matrix.
The LSP and LQUP variants of~\cite{IMH:1982} are used to reduce the
complexity rank deficient Gaussian elimination to that of
matrix multiplication.
Many other algorithmic decompositions exist allowing fraction free computations
\cite{Jeffrey:2010:lufact}, in-place computations~\cite{jgd:2008:toms,JPS:2013}
or sub-cubic rank-sensitive time
complexity~\cite{Storjohann:2000:thesis,JPS:2013}.  
In \cite{DPS:2013} we proposed a Gaussian elimination algorithm with a recursive
splitting of both row and column dimensions, and replacing row and column transpositions
by rotations. This elimination can compute simultaneously the row and
column rank profile while preserving the sub-cubic rank-sensitive time
complexity and keeping the computation in-place.

In this paper we first study the conditions a
\pluq decomposition algorithm must satisfy in order to reveal the rank
profile structure of a matrix. 
We  introduce in section~\ref{sec:rkp} the rank profile matrix $\RPM{A}$, a
normal form summarizing all rank profile information of a matrix and of
all its leading sub-matrices. We then decompose, in section~\ref{sec:structurePLUQ},
the pivoting strategy of any \pluq algorithm into two types of
operations: the search of the pivot and the permutation used to move it to the
main diagonal.  We propose a new search and a new permutation strategy and show
what rank profiles are computed using any possible combination of these
operations and the previously used searches and permutations.
In particular we show three new pivoting strategy combinations that compute the rank profile
matrix and use one of them, an iterative Crout \cupd with rotations,
to improve the base case and thus the overall performance of exact
Gaussian elimination.
Second, we show that preserving both the row and column rank profiles,
together with ensuring a monotonicity of the associated permutations,
allows us to compute faster several other matrix decompositions, such
as the LEU and Bruhat decompositions, and echelon forms.
In the following, $0_{m\times n}$ denotes the $m\times n$ zero matrix and
$A_{i..j,k..l}$ denotes the 
sub-matrix of $A$ of rows between $i$ and $j$ and columns between $k$ and $l$.
To a permutation $\sigma:\{1,\dots,n\}\rightarrow \{1,\dots,n\}$ we define the associated
permutation matrix $P_\sigma$, permuting rows by left multiplication: the rows
of $P_\sigma A$ are that of $A$ permuted by $\sigma$. Reciprocally, for a
permutation matrix $P$, we denote by $\sigma_P$ the associated permutation.

\section{The rank profile matrix}
\label{sec:rkp}
We start by introducing in Theorem~\ref{def:rankprofilematrix} the rank profile
matrix, that we will use throughout 
this document to summarize all information on the rank profiles of a matrix.
From now on, matrices are over a field $\K$ and a valid pivot is a
non-zero element.
\begin{definition} An $r$-sub-permutation matrix is a matrix of rank $r$ with only  $r$
  non-zero entries equal to one.
\end{definition}

\begin{lemma}
An $m\times n$ $r$-sub-permutation matrix
has at most one non-zero entry per row and per column,
and
can be written $P
  \begin{smatrix}
    I_r\\&0_{(m-r)\times (n-r)}
  \end{smatrix}Q$ where $P$ and $Q$ are permutation matrices.
\end{lemma}

\begin{theorem}\label{def:rankprofilematrix}
 Let $A\in\K^{m\times n}$. There exists a unique $m\times n$
 $\rank(A)$-sub-permutation matrix $\RPM{A}$ 
 of which every leading sub-matrix has the same
 rank as the corresponding
 leading sub-matrix of $A$. 
 This sub-per\-mu\-ta\-tion matrix is called the \em{rank profile matrix} of $A$.
\end{theorem}

\begin{proof}
  We prove existence by induction  on the row dimension of the leading
  submatrices.

  If $A_{1,1..n} = 0_{1\times n}$, setting $\RPM{}^{(1)} = 0_{1\times n}$ satisfies the
  defining condition.
  Otherwise, let $j$ be the index of the leftmost invertible element in
  $A_{1,1..n}$ and set $\RPM{}^{(1)}= e_j^T$ the j-th $n$-dimensional row canonical
  vector, which satisfies the defining condition.

  Now for a given $i\in\{1,\dots,m\}$, suppose that there is a unique $i\times n$
  rank profile matrix $\RPM{}^{(i)}$ such that $\rank(A_{1..i,1..j}) =
  \rank(\RPM{1..i,1..j})$ for every $j\in\{1..n\}$. 
  If $\rank(A_{1..i+1,1..n})=\rank(A_{1..i,1..n})$, then $\RPM{}^{(i+1)}=
  \begin{smatrix}
    \RPM{}^{(i)}\\0_{1\times n}
  \end{smatrix}$. 
  Otherwise, consider $k$, the smallest column index such that
  $\rank(A_{1..i+1,1..k})=\rank(A_{1..i,1..k})+1$ and set $\RPM{}^{(i+1)}=
  \begin{smatrix}
    \RPM{}^{(i)}\\e_k^T
  \end{smatrix}$. 
  Any leading sub-matrix of $\RPM{}^{(i+1)}$ has the same rank as the
  corresponding leading sub-matrix of $A$: 
first, for any leading subset of rows and columns with
less than $i$ rows, the case is covered by the induction; 
second 
define $
  \begin{bmatrix}
    B & u \\
    v^T & x 
  \end{bmatrix} = A_{1..i+1,1..k}$,
 where $u,v$ are vectors and $x$ is a scalar.
From the definition of $k$, $v$ is linearly dependent with $B$ and
thus any leading sub-matrix of $\begin{smatrix}B\\v^T\end{smatrix}$
has the same rank as the corresponding sub-matrix of
$\RPM{}^{(i+1)}$. Similarly, from the definition of $k$, the same
reasoning works when considering more than $k$ columns, with a rank
increment by $1$.\\
Lastly we show that $\RPM{}^{(i+1)}$ is a $r_{i+1}$-sub-per\-mu\-ta\-tion
matrix.
Indeed, $u$ is linearly dependent with the columns of~$B$: otherwise, 
$\rank(\begin{bmatrix}B&u\end{bmatrix})=\rank(B)+1$.
From the definition of $k$ we then have $\rank(\begin{smatrix}  B&u\\v^T&x\end{smatrix}) 
= rank(\begin{bmatrix}  B&u\end{bmatrix}) + 1 = \rank(B)+2 = \rank(
\begin{smatrix}  B\\v^T\end{smatrix})+2$ which is a contradiction.
Consequently, the $k$-th column of
$\RPM{}^{(i)}$ is all zero, and 
$\RPM{}^{(i+1)}$ is a $r$-sub-permutation matrix.

  To prove uniqueness, suppose there exist two distinct rank profile
  matrices $\RPM{}^{(1)}$ and $\RPM{}^{(2)}$ for a given matrix $A$ and let
  $(i,j)$ be some coordinates where $\RPM{1..i,1..j}^{(1)}\neq\RPM{1..i,1..j}^{(2)}$
  and $\RPM{1..i-1,1..j-1}^{(1)}=\RPM{1..i-1,1..j-1}^{(2)}$.
  Then, $\rank(A_{1..i,1..j})=\rank(\RPM{1..i,1..j}^{(1)}
)\neq
    \rank(\RPM{1..i,1..j}^{(2)}) = \rank(A_{1..i,1..j})$ which is a contradiction.
\end{proof}

\begin{example}$A= 
  \begin{smatrix}
    2 & 0 & 3& 0\\
    1 & 0 & 0& 0\\
    0 & 0 & 4 &0\\
    0 & 2 & 0 &1\\
  \end{smatrix}
  $ has $\RPM{A} = 
  \begin{smatrix}
    1& 0& 0& 0\\
    0& 0& 1& 0 \\
    0& 0& 0& 0\\
    0& 1& 0& 0
  \end{smatrix}
  $ for rank profile matrix over $\mathbb{Q}$.
\end{example}

\begin{remark}
  The matrix $E$ introduced in Malaschonok's LEU
  decomposition~\cite[Theorem~1]{Malaschonok:2010}, is in fact the rank profile
  matrix. There, the existence of this decomposition was only shown for $m=n=2^k$, and no
  connection was made to the relation with ranks and rank profiles.
  This connection was made in~\cite[Corollary~1]{DPS:2013}, and the existence of
  $E$ generalized to arbitrary dimensions $m$ and $n$. Finally, after proving its
  uniqueness here, we propose this definition as a new matrix normal form.
 \end{remark}

The rank profile matrix has the following properties:
\begin{lemma}\label{lem:rpm:prop} Let $A$ be a matrix.
\newcounter{myenum}
  \begin{compactenum}
  \item   $\RPM{A}$ is {\em diagonal} if $A$ has {\em generic rank profile}.
  \item $\RPM{A}$ is a {\em permutation} matrix if $A$ is  {\em invertible} 
  \item $\RRP(A) = \RS{\RPM{A}}$; $\CRP(A) = \CS{\RPM{A}}$.
\setcounter{myenum}{\theenumi}
\end{compactenum}
Moreover, for all $1\leq i\leq m$ and $1\leq j\leq n$, we have:
\begin{compactenum}\setcounter{enumi}{\themyenum}
  \item $\RRP(A_{1..i,1..j}) = \RS{(\RPM{A})_{1..i,1..j}}$ 
  \item $\CRP(A_{1..i,1..j}) = \CS{(\RPM{A})_{1..i,1..j}}$,
  \end{compactenum}
\end{lemma}
These properties show how to recover the row and column rank profiles
of $A$ and of any of its leading sub-matrix.

\section{Ingredients of a \pluq decomposition algorithm}
\label{sec:structurePLUQ}
Over a field, the LU 
decomposition generalizes to matrices with arbitrary rank
profiles, using row and column permutations   (in some cases such as
the CUP, or LSP decompositions, the row permutation is embedded in the structure of the $C$ or $S$
matrices). 
However such \pluq decompositions are not unique and not all of them will
necessarily reveal rank profiles and echelon forms.
We will characterize the conditions for a \pluq decomposition algorithm to
reveal the row or column rank profile or the rank profile matrix.

We consider the four  types of operations of a Gaussian elimination
algorithm in the processing of the $k$-th pivot:
\begin{compactdesc}
  \item[Pivot search:] finding an element to be used as a pivot,
  \item[Pivot permutation:] moving the pivot in diagonal position $(k,k)$  by column and/or row permutations,
  \item[Update:] applying the elimination at position $(i,j)$: 
\\$a_{i,j}\leftarrow a_{i,j} -a_{i,k}a_{k,k}^{-1}a_{k,j}$,
  \item[Normalization:] dividing the $k$-th row (resp.\ column) by the pivot.
\end{compactdesc}
Choosing how each of these operation is done, and when they are scheduled
results in an elimination algorithm.  Conversely, any Gaussian elimination
algorithm computing a \pluq decomposition can be viewed as a set of
specializations of each of these operations together with a scheduling. 

The choice of doing the normalization on rows or columns only determines which
of $U$ or $L$ will be unit triangular. The scheduling of the updates vary
depending on the type of algorithm used: iterative, recursive, slab or tiled block splitting,
with right-looking, left-looking or Crout variants~\cite{DDSV98}.
Neither the normalization nor the update impact the capacity to reveal rank
profiles and we will thus now focus on the pivot search and the permutations. 

Choosing a search and a permutation strategy  fixes the
 matrices $P$ and  $Q$ of the \pluq decomposition obtained and, as we
will see, determines the ability to recover information on the rank profiles.
Once these  matrices are fixed, the $L$ and the $U$ factors are unique.
We introduce the pivoting matrix.%
\begin{definition}\label{def:PivMat}
The pivoting matrix of a \pluq decomposition $A=PLUQ$ of rank $r$ is the $r$-sub-permutation matrix
$$\Pi_{P,Q}=P
\begin{bmatrix}
  I_r\\&0_{(m-r)\times (n-r)}
\end{bmatrix}
Q
.$$
\end{definition}
The $r$ non-zero elements of $\Pi_{P,Q}$ are  located at the initial positions of
the pivots in the matrix $A$. Thus $\Pi_{P,Q}$ summarizes the choices
made in the search and  permutation operations.

\begin{paragraph}{Pivot search}
The search operation vastly differs depending on the field of application. In
numerical dense linear algebra, numerical stability is the
main criterion for the selection of the pivot. In sparse linear algebra, the pivot
is chosen so as to reduce the fill-in produced by the update operation.
In order to reveal some information on the rank profiles, a notion of precedence
has to be used: a usual way to compute the row rank profile is to search in a
given row for a pivot and only move to the next row if the current row was found to be all
zeros. This guarantees that
each pivot will be on the first linearly independent row, and therefore the row
support of $\Pi_{P,Q}$ will be the row rank profile.
The precedence here is that the pivot's coordinates must minimize the order
for the first coordinate (the row index). 
As a generalization, we consider the most common preorders of the cartesian product $\{1,\ldots
m\}\times \{1,\ldots n\}$ inherited from the natural orders of each of its
components and describe the corresponding search strategies, 
minimizing this preorder: 
\begin{compactdesc}
  \item[Row order:] $(i_1,j_1)\preceq_{\text{row}} (i_2,j_2)$ iff $i_1\leq i_2$:
    {\em search for any invertible element in the first non-zero row.}
  \item[Column order:] $(i_1,j_1)\preceq_{\text{col}} (i_2,j_2)$ iff $j_1\leq
    j_2$. 
    {\em search for any invertible element in the first non-zero column.}
  \item[Lexicographic order:] $(i_1,j_1)\preceq_{\text{lex}} (i_2,j_2)$ iff $i_1<i_2$ or $i_1=i_2$ and
    $j_1 \leq j_2$:
    {\em search for the leftmost non-zero element of the first non-zero row.}
  \item[Reverse lexicographic order:] $(i_1,j_1)\preceq_{\text{revlex}} (i_2,j_2)$ iff $j_1<j_2$ or
    $j_1=j_2$ and $i_1 \leq i_2$: {\em search for the topmost  non-zero element
      of the first non-zero column.}
  \item[Product order:]\index{product order} $(i_1,j_1)\preceq_{\text{prod}} (i_2,j_2)$ iff
    $i_1\leq i_2$ and $j_1\leq j_2$:
    {\em search for any non-zero element at position $(i,j)$ being the
      only non-zero of the leading $(i,j)$ sub-matrix.}
\end{compactdesc}
\begin{example}
Consider the matrix
$
\begin{smatrix}
  0 & 0 & 0 & a & b\\
  0 & c & d & e & f\\
  g & h & i & j & k\\
  l & m & n & o & p
\end{smatrix}$, where each literal  is a non-zero element.
The minimum non-zero elements for each preorder are the following:
\begin{center}
  \begin{tabular}{ll}
\toprule
Row order & $a,b$ \\
Column order & $g,l$\\
Lexicographic order & $a$\\
Reverse lexic. order &$g$\\
Product order & $a,c,g$\\
\bottomrule
\end{tabular}

\end{center}
\end{example}

\end{paragraph}

\begin{paragraph}{Pivot permutation}

The pivot permutation moves a pivot from its initial position to the
leading diagonal. Besides this constraint all possible choices are left for the
remaining values of the permutation. 
Most often, it is done by row or column transpositions, as it clearly involves
a small amount of data movement.
However, %
these transpositions can break the precedence relations in the set of rows or
columns, and can therefore prevent the recovery of the rank profile information.
A pivot permutation that leaves the precedence relations unchanged will be
called  $k$-monotonically increasing.
\begin{definition}
  A permutation of $\sigma \in \mathcal{S}_n$ is called
  $k$-mono\-ton\-i\-cal\-ly increasing if its last $n-k$ values 
  form a monotonically increasing sequence.
\end{definition}
In particular, the last $n-k$ rows of the associated
row-permutation matrix $P_\sigma$
are in row echelon form.
For example, the cyclic shift between indices $k$ and $i$, with $k<i$
defined as $R_{k,i}=(1,\ldots,k-1,i,k,k+1,\ldots,i-1,i+1,\ldots,n)$, that we will call a
$(k,i)$-rotation, is an elementary $k$-monotonically increasing permutation.
\begin{example} The $(1,4)$-rotation $R_{1,4}=(4,1,2,3)$ is a
  $1$-mono\-to\-ni\-cal\-ly increasing permutation. Its row permutation matrix
  is 
$\begin{smatrix}
    0& & & 1\\
    1&    &   & \\
    &1& & \\
    & & 1&0\\
  \end{smatrix}$. In fact, any $(k,i)$-rotation is a $k$-monotonically
  increasing permutation. 
\end{example}

Monotonically increasing permutations can be composed as stated in
Lemma~\ref{lem:permcompo}. 
\begin{lemma}\label{lem:permcompo}
If $\sigma_1 \in \mathcal{S}_n$ is a $k_1$-monotonically increasing permutation
and $\sigma_2\in \mathcal{S}_{k_1} \times \mathcal{S}_{n-k_1}$ a
$k_2$-monotonically increasing permutation with $k_1<k_2$ then the
permutation $\sigma_2 \circ \sigma_1$ is a $k_2$-monotonically increasing
permutation. 
\end{lemma}

 \begin{proof}
   The last $n-k_2$ values of $\sigma_2 \circ \sigma_1$ are the image of a
   sub-sequence of $n-k_2$ values from the last $n-k_1$ values of $\sigma_1$
   through the  monotonically increasing function~$\sigma_2$.
 \end{proof}

Therefore an iterative algorithm, using rotations as elementary pivot
permutations, maintains the property that the permutation matrices $P$ and $Q$ at
any step $k$ are $k$-monotonically increasing. A similar property also applies
with recursive algorithms. 
\end{paragraph}

\section{How to reveal rank profiles}
\label{sec:cond}

A PLUQ decomposition reveals the row (resp.\ column) rank profile if it can be
read from the first $r$ values of the permutation matrix $P$ (resp.\ $Q$).
Equivalently, by Lemma~\ref{lem:rpm:prop}, this means that the row (resp.\ column) support of the pivoting
matrix $\Pi_{P,Q}$ equals that of the rank profile matrix.

\begin{definition}
  The decomposition $A=PLUQ$ reveals:
  \begin{compactenum}
  \item the row rank profile if $\RS{\Pi_{P,Q}}= \RS{\RPM{A}}$,
  \item the col. rank profile if $\CS{\Pi_{P,Q}}= \CS{\RPM{A}}$,
  \item the rank profile matrix if $\Pi_{P,Q}=\RPM{A}$.
  \end{compactenum}
\end{definition}

\begin{example}\label{ex:rankprofile}
$A= 
  \begin{smatrix}
    2 & 0 & 3& 0\\
    1 & 0 & 0& 0\\
    0 & 0 & 4 &0\\
    0 & 2 & 0 &1\\
  \end{smatrix}
  $ 
  has $\RPM{A} = 
  \begin{smatrix}
    1& 0& 0& 0\\
    0& 0& 1& 0 \\
    0& 0& 0& 0\\
    0& 1& 0& 0
  \end{smatrix}
  $ for rank profile matrix over $\mathbb{Q}$.
Now the pivoting matrix obtained from a \pluq decomposition with a pivot search
operation following the row order (any column, first non-zero row) could be the matrix
  $\Pi_{P,Q} = 
\begin{smatrix}
  0&0&1&0\\
  1&0&0&0\\
  0&0&0&0\\
  0&1&0&0\\
\end{smatrix}
$. As these matrices share the same row support, the matrix $\Pi_{P,Q}$ reveals the row
rank profile of $A$.
\end{example}

\begin{remark}\label{rem:SwapsConterex}
  Example~\ref{ex:rankprofile}, suggests that a pivot search strategy 
  minimizing row and column indices could be a sufficient condition to recover both row
  and column rank profiles at the same time, regardless the pivot permutation.
  However, this is unfortunately not the case. Consider for
  example a search based on  the lexicographic order (first non-zero column of
  the first non-zero row) with transposition permutations, run on the matrix:
  $A= 
  \begin{smatrix}
    0 & 0 & 1\\
    2 & 3 & 0\\
  \end{smatrix}$. Its rank profile matrix is $\RPM{A} = 
  \begin{smatrix}
    0&0&1\\
    1&0&0
  \end{smatrix}
  $ whereas the pivoting matrix could be 
$
  \Pi_{P,Q}=\begin{smatrix}
    0&0&1\\
    0&1&0
  \end{smatrix}
  $, which does not reveal the column rank profile.
 This is due to the fact that the column transposition performed for the
  first pivot changes the order in which the columns will be inspected in
  the search for the second pivot. 
\end{remark}

We will show that if the pivot permutations preserve the order in which the still
unprocessed columns or rows appear, then the pivoting matrix will equal the rank
profile matrix. This is achieved by the monotonically increasing permutations. 

Theorem~\ref{th:RPandPerm} shows how the ability of a \pluq
decomposition algorithm to recover the rank profile information relates to the
use of monotonically increasing permutations.
More precisely, it considers an arbitrary step in a PLUQ decomposition where $k$
pivots have been found in the elimination of an $\ell\times p$ leading sub-matrix $A_1$
of the input matrix $A$.

\begin{theorem}
\label{th:RPandPerm}
  Consider a partial \pluq decomposition of an $m\times n$ matrix $A$:
\[
A = P_1
\begin{bmatrix}
  L_1 \\ M_1 & I_{m-k}
\end{bmatrix}
\begin{bmatrix}
  U_1 & V_1\\
      & H
\end{bmatrix}
Q_1
\] where $
\begin{bmatrix} L_1\\M_1\end{bmatrix}$ is $m\times k$ lower triangular and 
$\begin{bmatrix}  U_1 & V_1\end{bmatrix}$ is $ k\times n$ upper triangular,
and  let $A_1$
be some $\ell \times p$ leading sub-matrix  of $A$, for $\ell,p\geq k$.
Let $H=P_2L_2U_2Q_2$ be a \pluq decomposition of $H$.
Consider the \pluq decomposition
\[
A=\underbrace{P_1
\begin{bmatrix}
  I_k\\& P_2
\end{bmatrix}}_{P}
\underbrace{
\begin{bmatrix}
  L_1\\P_2^TM_1&L_2
\end{bmatrix}}_L
\underbrace{\begin{bmatrix}
  U_1&V_1Q_2^T\\
  &U_2
\end{bmatrix}}_{U}
\underbrace{\begin{bmatrix}
  I_k\\& Q_2
\end{bmatrix}
Q_1}_{Q}.
\]

Consider the following clauses:
  \begin{compactenum}[(i)]
  \item $\RRP(A_1) = \RS{\Pi_{P_1,Q_1}}$ \label{clause:rrp1}
  \item $\CRP(A_1) = \CS{\Pi_{P_1,Q_1}}$ \label{clause:crp1}
  \item $\RPM{A_1} = \Pi_{P_1,Q_1} $ \label{clause:rpm1}
  \item $\RRP(H) = \RS{\Pi_{P_2,Q_2}}$ \label{clause:rrp2}
  \item $\CRP(H) = \CS{\Pi_{P_2,Q_2}} $ \label{clause:crp2}
  \item $\RPM{H} = \Pi_{P_2,Q_2}$ \label{clause:rpm2}
  \end{compactenum}
\begin{compactenum}[(i)]
 \setcounter{enumi}{6}
\item $P_1^T$ is $k$-monotonically increasing or ($P_1^T$ is $\ell$-mono\-tonically
    increasing and $p=n$)\label{clause:PMI} \label{clause:Pmonotinc}
\item $Q_1^T$ is $k$-monotonically increasing or ($Q_1^T$ is $p$-mono\-tonically
    increasing and $\ell=m$)\label{clause:Qmonotinc}
\end{compactenum}
Then,
\begin{compactenum}[(a)]
\item if (\ref{clause:rrp1}) or (\ref{clause:crp1}) or (\ref{clause:rpm1}) then $H= \begin{bmatrix} 0_{(\ell-k)\times(p-k)}&*\\ *&* \end{bmatrix}$\label{th:H}
\item if (\ref{clause:Pmonotinc})
then ((\ref{clause:rrp1}) and (\ref{clause:rrp2})) $\Rightarrow \RRP(A)=\RS{\Pi_{P,Q}} $;\label{th:RP:row}

\item if (\ref{clause:Qmonotinc})  
then 
((\ref{clause:crp1}) and (\ref{clause:crp2})) $\Rightarrow \CRP(A) = \CS{\Pi_{P,Q}} $;\label{th:RP:col}

\item if (\ref{clause:Pmonotinc}) and (\ref{clause:Qmonotinc})  then (\ref{clause:rpm1}) and (\ref{clause:rpm2}) $\Rightarrow  \RPM{A}=\Pi_{P,Q}$.\label{th:RP:both}
\end{compactenum}

\end{theorem}
\begin{proof}
Let $P_1=
\begin{bmatrix}
  P_{11} & E_1
\end{bmatrix}$ and $Q_1=
\begin{bmatrix}
  Q_{11}\\F_1
\end{bmatrix}$
where  $E_1$ is $m\times(m-k)$ and  $F_1$ is $(n-k)\times n$. 
On one hand we have
\begin{eqnarray}
  A &=& 
\underbrace{\begin{bmatrix} P_{11}&E_1\end{bmatrix}
  \begin{bmatrix} L_1\\M_1 \end{bmatrix}
  \begin{bmatrix} U_1&V_1 \end{bmatrix}
  \begin{bmatrix} Q_{11}\\F_1\end{bmatrix}}_{B} + 
    E_1HF_1. \label{eq:RP:PLUQ}
\end{eqnarray}

On the other hand,
  \begin{eqnarray}
    \Pi_{P,Q} &=& P_1 
    \begin{bmatrix}
      I_k\\&P_2
    \end{bmatrix}
    \begin{bmatrix}
      I_r\\&0_{(m-r)\times(n-r)}
    \end{bmatrix}
    \begin{bmatrix}
      I_k\\&Q_2
    \end{bmatrix}Q_1\notag \\
    &=& P_1 
    \begin{bmatrix}
      I_k\\&\Pi_{P_2,Q_2}
    \end{bmatrix}
    Q_1 \notag = \Pi_{P_1,Q_1} + E_1\Pi_{P_2,Q_2}F_1.
    \label{eq:RP:pi}
  \end{eqnarray}

Let $\LA = 
\begin{bmatrix}
  A_1&0\\
  0&0_{(m-\ell)\times(n-p)}
\end{bmatrix}
$ and denote by $B_1$ the $\ell\times p$ leading sub-matrix of~$B$.

\begin{compactenum}[(a)]
\item 
The clause \eqref{clause:rrp1} or \eqref{clause:crp1} or \eqref{clause:rpm1} implies that all $k$ pivots of the partial
  elimination were found within the $\ell\times p$ sub-matrix
  $A_1$. Hence $\rank(A_1)=k$ and we can write 
$P_1=\begin{bmatrix}\begin{array}{c} P_{11}\\0_{(m-\ell)\times k}\end{array}&E_1\end{bmatrix}$
 and $Q_1=\begin{bmatrix} Q_{11}& 0_{k\times(n-p)}\\ \multicolumn{2}{c}{F_1}\end{bmatrix}$,
 and the matrix $A_1$ writes 
 \begin{equation}\label{eq:A1B1}
A_1 = \begin{bmatrix}I_\ell &0 \end{bmatrix}
A
\begin{smatrix}I_p\\0 \end{smatrix} 
= 
B_1 +
 \begin{bmatrix} I_\ell&0 \end{bmatrix}E_1HF_1\begin{smatrix} I_p\\0 \end{smatrix}.
\end{equation}
Now $\rank(B_1)=k$ as a sub-matrix of $B$ of rank $k$ and since
\begin{eqnarray*}
B_1        &=& 
\begin{bmatrix}P_{11}&  \begin{bmatrix}  I_\ell&0 \end{bmatrix}\cdot E_1\end{bmatrix}
\begin{bmatrix} L_1\\M_1\end{bmatrix}
\begin{bmatrix} U_1&V_1\end{bmatrix}
\begin{bmatrix} Q_{11}\\F_1 \cdot \begin{smatrix} I_p\\0 \end{smatrix}\end{bmatrix}  \\
&=&  P_{11}L_1U_1Q_{11} + 
        \begin{bmatrix} I_\ell&0 \end{bmatrix} E_1M_1
        \begin{bmatrix}U_1&V_1\end{bmatrix} Q_1 \begin{smatrix} I_p\\0 \end{smatrix}
  \end{eqnarray*}
where the first term, $P_{11}L_1U_1Q_{11}$, has rank $k$ and the second term has a
disjoint row support.

Finally, consider the term   $\begin{bmatrix}
  I_\ell&0 \end{bmatrix}E_1HF_1\begin{smatrix} I_p\\0 \end{smatrix}$ of
equation~\eqref{eq:A1B1}.  As its row
 support  is disjoint with that of the pivot rows of $B_1$, it has to
be composed of rows linearly dependent with the pivot rows of $B_1$ to ensure
that $\rank(A_1)=k$. As its
column support is disjoint with that of the pivot columns of $B_1$, we conclude
that it must be the zero matrix.
Therefore the
leading $(\ell-k)\times (p-k)$ sub-matrix of $E_1HF_1$ is zero. 
\item From~(\ref{th:H}) we know that $A_1= B_1$. Thus
  $\RRP(B) = \RRP(A_1)$. Recall that $A=B+E_1HF_1$.
No pivot row of $B$ can be made linearly dependent by adding rows of $E_1HF_1$,
as the column position of the pivot is always zero in the latter
matrix. For the same reason, no pivot row of $E_1HF_1$ can be made linearly
dependent by adding rows of $B$.
From~\eqref{clause:rrp1}, the set of pivot rows of $B$ is $\RRP(A_1)$,
which shows that 
\begin{equation}
\RRP(A)=\RRP(A_1)\cup \RRP(E_1HF_1).
\label{eq:rrpa}
\end{equation}

  Let $\sigma_{E_1}:\{1..m-k\}\rightarrow \{1..m\}$ be the map representing the
  sub-permutation $E_1$ (i.e. such that $E_1[\sigma_{E_1}(i),i]=1 \ \forall i$).
  If $P_1^T$ is $k$-monotonically increasing, the matrix $E_1$ has full column
  rank and is in column echelon  form, which implies that 
\begin{eqnarray}
\RRP(E_1HF_1) &=& \sigma_{E_1} (\RRP(HF_1))\notag\\
             &=& \sigma_{E_1}(\RRP(H)),\label{eq:rrpehf}
\end{eqnarray} 
since  $F_1$ has full row rank.
If $P_1^T$ is $\ell$ monotonically increasing, we can write $E_1=
\begin{bmatrix} E_{11}&E_{12}\end{bmatrix}$, where the $m\times (m-\ell)$ matrix
$E_{12}$ is in column echelon form. If $p=n$, the matrix $H$ writes $H=
\begin{bmatrix}  0_{(\ell-k)\times (n-k)}\\H_2 \end{bmatrix}$. Hence we have
$E_1HF_1 = E_{12}H_2F_1$ which also implies
\[
\RRP(E_1HF_1)  = \sigma_{E_1}(\RRP(H)).
\]
  From equation~(\ref{eq:RP:pi}), the row
  support of $\Pi_{P,Q}$ is that of $\Pi_{P_1,Q_1} + E_1\Pi_{P_2,Q_2}F_1$, which is the
  union of the row support of these two terms as they are disjoint. Under the
  conditions of point~\eqref{th:RP:row}, this row support is the union of
  $\RRP(A_1)$ and $\sigma_{E_1}(\RRP(H))$, which is, from~\eqref{eq:rrpehf}
  and~\eqref{eq:rrpa}, $\RRP(A)$.
\item Similarly as for point~(\ref{th:RP:row}).
\item From~(\ref{th:H}) we have still $A_1=B_1$. 
  Now since $\rank(B)=\rank(B_1)=\rank(A_1)=k$, there is no other non-zero element in $\RPM{B}$
  than those in $\RPM{\LA}$ and $\RPM{B}= \RPM{\LA}$.
  The row and column support of  $\RPM{B}$ and that of $E_1HF_1$ are disjoint. Hence 
\begin{equation}\label{eq:RP:RA}
\RPM{A} =    \RPM{\LA} + \RPM{E_1HF_1}.
\end{equation}
If both $P_1^T$ and $Q_1^T$ are $k$-monotonically increasing, the matrix $E_1$ is in
  column echelon form and the matrix $F_1$ in row echelon form. Consequently, the matrix
  $E_1HF_1$ is a copy of the matrix $H$ with $k$ zero-rows and $k$ zero-columns
  interleaved, which does not impact the linear dependency relations between 
  the non-zero rows and columns. As a consequence 
\begin{equation}
\RPM{E_1HF_1} =  E_1\RPM{H}F_1. \label{eq:EHFRPM}
\end{equation}
 Now if $Q_1^T$ is $k$-monotonically increasing, $P_1^T$ is
  $\ell$-mono\-ton\-ically increasing and $p=n$, then, using notations of
  point~\eqref{th:RP:row}, $E_1HF_1 = E_{12}H_2F_1$ where
  $E_{12}$ is in column echelon form. Thus  $\RPM{E_1HF_1} =  E_1\RPM{H}F_1$ for
  the same reason. The symmetric case where $Q_1^T$ is $p$-monotonically
  increasing and $\ell=m$ works similarly.
  Combining equations~(\ref{eq:RP:pi}),~(\ref{eq:RP:RA}) and~(\ref{eq:EHFRPM})  gives 
  $\RPM{A}  = \Pi_{P,Q}$.%
\end{compactenum}%
\vspace{-1em}
\, \hfill\, \end{proof}

\section{Algorithms for rank profiles}\label{sec:algo}
Using Theorem~\ref{th:RPandPerm}, we  deduce what rank profile information is
revealed by a PLUQ algorithm by the way the Search and the Permutation
operations are done.
Table~\ref{tab:RPRSearchPerm} summarizes these results, and points to instances
known in the literature, implementing the corresponding type of elimination.
More precisely, we first distinguish in this table the ability to compute the
row or column rank profile or the rank profile matrix, but we also indicate
whether the resulting PLUQ decomposition preserves the monotonicity of the rows
or columns. Indeed some algorithm may compute the rank profile matrix,
but break the precedence relation between the linearly dependent rows or
columns,  making it unusable as a base case for a block algorithm of higher level.

\begin{table*}[htb]\small
\begin{center}
\begin{tabular}{llllll}
\toprule
\textbf{Search }&\textbf{Row
  Perm.}&\textbf{Col. Perm.}&\textbf{Reveals}&\textbf{Monotonicity} & \textbf{Instance}\\
\midrule
 Row order  & Transposition & Transposition &\RRP & &\cite{IMH:1982, JPS:2013}\\ 
 Col. order  & Transposition & Transposition &\CRP & & \cite{KG:1985, JPS:2013}\\ 
\midrule
  \multirow{3}{*}{Lexicographic} & Transposition & Transposition & \RRP & &\cite{Storjohann:2000:thesis}\\
  & Transposition &Rotation & \RRP, \CRP, \RPM{}&Col. &here \\
  & Rotation &Rotation & \RRP, \CRP, \RPM{}&Row, Col. &here \\
\midrule
\multirow{3}{*}{Rev. lexico.} & Transposition & Transposition & \CRP & &\cite{Storjohann:2000:thesis}\\
   & Rotation & Transposition & \RRP, \CRP, \RPM{}& Row& here\\
   & Rotation & Rotation & \RRP, \CRP, \RPM{}& Row, Col.& here\\
\midrule
  \multirow{3}{*}{Product}  & Rotation &Transposition &\RRP& Row &here \\
    & Transposition &Rotation &\CRP& Col &here \\
    & Rotation &Rotation &\RRP, \CRP, \RPM{}& Row, Col.&\cite{DPS:2013} \\
\bottomrule
\end{tabular}
\caption{Pivoting Strategies revealing rank profiles}\label{tab:RPRSearchPerm}
\end{center}
\vspace{-2em}
\end{table*}

\subsection{Iterative algorithms}
We start with iterative algorithms, where each iteration handles one pivot at a
time. Here Theorem~\ref{th:RPandPerm} is applied with $k=1$, and
the partial elimination represents how one pivot is being treated. The
elimination of $H$ is done by induction. 

\paragraph{Row and Column order Search}

The row order pivot search operation is of the form: 
\textit{any non-zero element in the first non-zero row}. 
Each row is inspected in order, and a new row is considered only
when the previous row is all zeros.  
With the notations of Theorem~\ref{th:RPandPerm}, this means that $A_1$ is the
leading $\ell\times n$ sub-matrix of $A$, where $\ell$ is the index of 
the first non-zero row of $A$.
When permutations $P_1$ and $Q_1$, moving the pivot from
position $(\ell,j)$ to $(k,k)$ are transpositions, 
the matrix $\Pi_{P_1,Q_1}$ is the element $E_{\ell,j}$ of the canonical basis. 
Its row rank profile is $(\ell)$ which is that of the $\ell
\times n$ leading sub-matrix $A_1$. Finally, the permutation $P_1$ is
$\ell$-monotonically increasing, and Theorem~\ref{th:RPandPerm}
case~(\ref{th:RP:row}) can be applied to prove by induction that any such
algorithm will reveal the row rank profile: $\RRP(A)=\RS{\Pi_{P,Q}}$.
The case of the column order search is similar. %
\paragraph{Lexicographic order based pivot search}

In this case the Pivot Search operation is of the form: 
\textit{first non-zero element in the first non-zero row}. 
The lexicographic order being compatible with the row order, the above results
hold when transpositions are used and the row rank profile is revealed. If in
addition column rotations are used, $Q_1=R_{1,j}$ which is $1$-monotonically
increasing. Now $\Pi_{P_1,Q_1}=E_{\ell,j}$ which is the rank profile matrix of
the $\ell\times n$ leading sub-matrix $A_1$ of $A$. Theorem~\ref{th:RPandPerm}
case~(\ref{th:RP:both}) can be applied to prove by induction that any such
algorithm will reveal the rank profile matrix: $\RPM{A}=\Pi_{P,Q}$. 
Lastly, the use of row rotations, ensures that the order of the linearly
dependent rows will be preserved as well.
Algorithm~\ref{alg:pluq:iter} is an instance of Gaussian elimination with a
lexicographic order search and rotations for row and column permutations.

The case of the reverse lexicographic order search is similar.
As an example, the algorithm in~\cite[Algorithm 2.14]{Storjohann:2000:thesis} is based
on a reverse lexicographic order search but with transpositions for the row
permutations. Hence it only reveals the column rank profile.

\paragraph{Product order based pivot search}

The search here consists in finding any non-zero element $A_{\ell,p}$ such that
the $\ell\times p$ leading sub-matrix $A_1$ of $A$ is all zeros except this
coefficient. If the row and column permutations are the rotations $R_{1,\ell}$
and $R_{1,p}$, we have
$\Pi_{P_1,Q_1}=E_{\ell,p}=\RPM{A_1}$. Theorem~\ref{th:RPandPerm}
case~(\ref{th:RP:both}) can be applied to prove by induction that any such
algorithm will reveal the rank profile matrix: $\RPM{A}=\Pi_{P,Q}$. 
An instance  of such an algorithm is given in~\cite[Algorithm~2]{DPS:2013}.
If $P_1$ (resp.\ $Q_1$) is a transposition, then 
Theorem~\ref{th:RPandPerm} case~(\ref{th:RP:col}) (resp.\ case~(\ref{th:RP:row}))
applies to show by induction that the columns (resp.\ row) rank profile is revealed.

\subsection{Recursive algorithms}

A recursive Gaussian elimination algorithm can either split one of the row or
column dimension, cutting the matrix in  wide or tall rectangular slabs, or
split both dimensions, leading to a decomposition into tiles.

\begin{paragraph}  {Slab recursive algorihtms}
Most algorithms computing rank profiles are slab
recursive~\cite{IMH:1982, KG:1985, Storjohann:2000:thesis, JPS:2013}.
When the row dimension is split, this means that the search space for pivots is
the whole set of columns, and Theorem~\ref{th:RPandPerm} applies with
$p=n$. This corresponds to a either a row or a lexicographic order.
From case(~\ref{th:RP:row}), one shows that, with transpositions, the algorithm
recovers the row rank profile, provided that the base case does.
If in addition, the elementary column permutations are rotations, then
case~(\ref{th:RP:both}) applies and the rank profile matrix is recovered.
Finally, if rows are also permuted by monotonically increasing permutations,
then the PLUQ decomposition also respects the monotonicity of the linearly
dependent rows and columns.
The same reasoning holds when splitting the column dimension.
\end{paragraph}

\begin{paragraph}  {Tile recursive algorithms}

Tile recursive Gaussian elimination
algorithms~\cite{DPS:2013,Malaschonok:2010,DuRo:2002} are more involved, 
especially when dealing with rank deficiencies, and we refer to~\cite{DPS:2013}
for a detailed description of such an algorithm.
Here, the search area $A_1$ has arbitrary dimensions $\ell\times p$, often
specialized as $m/2 \times n/2$. As a consequence, the pivot search can not
satisfy neither row, column, lexicographic or reverse lexicographic orders.
Now, if the pivots selected in the elimination of  $A_1$ minimizes the product
order, then they necessarily also respect this order as pivots of the whole
matrix $A$. Now, from~(\ref{th:H}), the remaining matrix $H$ writes
$H=\begin{bmatrix}  0_{(\ell-k)\times(p-k)} & H_{12}\\H_{21} &
  H_{22}\end{bmatrix}$ and its elimination is done by two independent eliminations
on the blocks $H_{12}$ and $H_{21}$, followed by some update of $H_{22}$ and a
last elimination on it. Here again, pivots minimizing the row order on $H_{21}$
and $H_{12}$ are also pivots minimizing this order for $H$, and so are those of
the fourth elimination. Now the block row and column permutations used
in~\cite[Algorithm~1]{DPS:2013} to form the PLUQ decomposition are
$r$-monotonically increasing. Hence, from case~(\ref{th:RP:both}), the algorithm
computes the rank profile matrix and preserves the monotonicity.
If only one of the row or column permutations are rotations, then
case~(\ref{th:RP:row}) or~(\ref{th:RP:col}) applies to show that either the row or
the column rank profile is computed.

\end{paragraph}

\section{Rank profile matrix based triangularizations}
\subsection{LEU decomposition}

The LEU decomposition introduced in~\cite{Malaschonok:2010} involves a lower
triangular matrix $L$, an upper triangular matrix $U$ and a $r$-sub-permutation
matrix $E$.

\begin{theorem}
  Let $A=PLUQ$ be a PLUQ decomposition revealing the rank profile matrix
  ($\Pi_{P,Q}=\RPM{A}$). Then an LEU decomposition of $A$ with $E=\RPM{A}$
  is obtained as follows (only using row and column permutations):
\begin{equation}\label{eq:LEU}
 A = \underbrace{P\begin{bmatrix} L&0_{m\times (n-r)}\end{bmatrix}P^T}_{\overline{L}} 
 \underbrace{P\begin{bmatrix}I_r\\&0\end{bmatrix}Q}_{E} 
 \underbrace{Q^T\begin{bmatrix} U\\0_{(n-r)\times n}\end{bmatrix}Q}_{\overline{U}}
\end{equation}
 \end{theorem}
 \begin{proof}
First $E=P\begin{smatrix}I_r\\&0\end{smatrix}Q = \Pi_{P,Q} = \RPM{A}$.
Then there only needs to show that $\overline{L}$ is lower triangular and
   $\overline{U}$ is upper triangular.
Suppose that $\overline{L}$ is not lower triangular, let $i$ be the first row
index such that $\overline{L}_{i,j}\neq0$ for some $i<j$. 
First $j\in\RRP(A)$ since the non-zero columns in $\overline{L}$ are placed
according to the first $r$ values of $P$.
Remarking that $A = P \begin{bmatrix} L&0_{m\times (n-r)}\end{bmatrix}
\begin{bmatrix}  \multicolumn{2}{c}{U}\\ 0 & I_{n-r} \end{bmatrix} Q$, and since
right multiplication by a non-singular matrix does not change row rank profiles,
we deduce that
$\RRP(\Pi_{P,Q})=\RRP(A)=\RRP(\overline{L})$.
If $i\notin\RRP(A)$, then the $i$-th row of $\overline{L}$ is linearly
dependent with the previous rows, but none of them has a non-zero element in
column $j>i$. Hence $i\in\RRP(A)$.

Let $(a,b)$ be the position of the coefficient $\overline{L}_{i,j}$ in $L$, that
is $a=\sigma_P^{-1}(i), b=\sigma_P^{-1}(j)$. Let also $s=\sigma_Q(a)$ and
$t=\sigma_Q(b)$ so that the pivots at diagonal position $a$ and $b$ in $L$
respectively correspond to ones in $\RPM{A}$ at positions $(i,s)$ and $(j,t)$.
Consider the $\ell\times p$ leading sub-matrices $A_1$ of $A$ where
$\ell=\max_{x=1..a-1}(\sigma_P(x))$ and
$p=\max_{x=1..a-1}(\sigma_Q(x))$.
On one hand $(j,t)$ is an index position in $A_1$ but not $(i,s)$, since
otherwise $\rank(A_1) = b$.
Therefore, $(i,s) \nprec_{prod} (j,t)$, and $s>t$ as $i<j$.
As coefficients $(j,t)$ and $(i,s)$ are pivots in $\RPM{A}$ and $i<j$ and $t<s$,
there can not be a non-zero element above $(j,t)$ at row $i$ when it is chosen
as a pivot. Hence $\overline{L}_{i,j}=0$ and $\overline{L}$ is lower triangular.
The same reasoning applies to show that $\overline{U}$ is upper triangular.
 \end{proof}

\begin{remark}\label{rem:cexleuuniq} 
Note that the LEU decomposition with $E=\RPM{A}$ is not unique, even for
invertible matrices. As a counter-example, the following decomposition holds for
any $a\in\K$:
\begin{equation}
\left[\begin{matrix} 0 & 1 \\ 1 & 0 \end{matrix}\right]
=
\left[\begin{matrix} 1 & 0 \\ a & 1 \end{matrix}\right]
\left[\begin{matrix} 0 & 1 \\ 1 & 0 \end{matrix}\right]
\left[\begin{matrix} 1 & -a \\ 0 & 1 \end{matrix}\right]
\end{equation}
\end{remark}

\subsection{Bruhat decomposition}
The Bruhat decomposition, that has inspired Malascho\-nok's LEU
decomposition~\cite{Malaschonok:2010}, is another decomposition with a central permutation matrix~\cite{Bourbaki:2008:lie,Grigoriev:1981:bruhat}.
\begin{theorem}[\cite{Bourbaki:2008:lie}]
Any invertible matrix $A$ can be written as $A=V P U$ for $V$ and $U$ uppper triangular invertible matrices and $P$ a permutation matrix. The latter decomposition is called the {\em Bruhat decomposition} of $A$.
\end{theorem}
It was then naturally extended to singular square matrices
by~\cite{Grigoriev:1981:bruhat}. 
Corollary~\ref{cor:bruhat} generalizes it to matrices with arbitrary dimensions,
and relates it to the PLUQ decomposition.
\begin{corollary}\label{cor:bruhat}
  Any $m\times n$ matrix of rank $r$ has a $VPU$ decomposition, where $V$ and $U$ are upper
  triangular matrices, and $P$ is a $r$-sub-permutation matrix.
\end{corollary}
\begin{proof}
Let $J_n$ be the unit anti-diagonal matrix. From the LEU decomposition of $J_nA$,
we have
  $A= \underbrace{J_nLJ_n}_V\underbrace{J_nE}_P U$ where $V$ is upper triangular.
\end{proof}

\subsection{Relation to LUP and PLU decompositions}

The \lup decomposition $A=LUP$ only exists for matrices with generic row rank profile
(including matrices with full row rank). Corollary~\ref{cor:elu} shows upon
which condition the permutation matrix $P$ equals the rank profile matrix $\RPM{A}$.
Note that although the rank profile $A$ is trivial in such cases, the matrix
$\RPM{A}$ still carries important information on the row and column rank
profiles of all leading sub-matrices of $A$.

\begin{corollary}\label{cor:elu}
  Let $A$ be an $m\times n$ matrix.

If $A$ has generic column rank profile, then any \plu decomposition $A=PLU$
    computed using reverse lexicographic order search and row rotations is such
    that 
    $\RPM{A}= P \begin{smatrix}I_r\\&0\end{smatrix}$. In particular,
    $P=\RPM{A}$ if $r=m$.

  If $A$ has generic row rank profile, then any \lup decomposition $A=LUP$
    computed using lexicographic order search and column rotations is such that 
    $\RPM{A}= \begin{smatrix}I_r \\ & 0\end{smatrix}P$. In particular,
    $P=\RPM{A}$ if $r=n$.
\end{corollary}

\begin{proof}
  Consider $A$ has generic column rank profile.
  From table~\ref{tab:RPRSearchPerm}, any \pluq decomposition algorithm with a reverse
  lexicographic order based search and rotation based row permutation is such
  that $\Pi_{P,Q}=P \begin{smatrix}  I_r\\& \end{smatrix}Q = \RPM{A}$. Since the
  search follows the reverse lexicographic order and the matrix has generic
  column rank profile, no column will be permuted in this elimination, and
  therefore $Q=I_n$.
  The same reasoning hold for when $A$ has generic row rank profile.
\end{proof}

Note that the $L$ and $U$ factors in a \plu decomposition are uniquely determined
by the permutation $P$. Hence, when the matrix has full row rank, $P=\RPM{A}$
and the decomposition $A=\RPM{A}LU$ is unique. Similarly the decomposition
$A=LU\RPM{A}$ is unique when the matrix has 
full column rank.
Now when the matrix is rank deficient with generic row rank profile, there is no
longer a unique \plu decomposition revealing the rank profile matrix: any
permutation applied to the last $m-r$ columns of $P$ and the last $m-r$ rows of
$L$ yields a \plu decomposition where $\RPM{A}=P \begin{smatrix} I_r\\& \end{smatrix}$. 

Lastly, we remark that the only situation where the rank profile matrix \RPM{A}
can be read directly as a sub-matrix of $P$ or $Q$ is as in
corollary~\ref{cor:elu}, when the matrix $A$ has generic row or column rank
profile. 
Consider a \pluq decomposition $A=PLUQ$ revealing the rank profile matrix
($\RPM{A}=P\begin{smatrix} I_r\\&\end{smatrix}Q$) such that $\RPM{A}$ is a
sub-matrix of $P$. This means that $P=\RPM{A}+S$ where $S$ has disjoint row and
column support with $\RPM{A}$.
We have
$  \RPM{A} = (\RPM{A}+S)  \begin{smatrix}    I_r\\&  \end{smatrix}  Q
          = (\RPM{A}+S)  \begin{smatrix}    Q_1\\0_{(n-r)\times
      n}  \end{smatrix}  
$.
Hence $\RPM{A}(I_n- \begin{smatrix}  Q_1\\0_{(n-r)\times n}  \end{smatrix}) =S \begin{smatrix}  Q_1\\0_{(n-r)\times n}  \end{smatrix}
$ but the row support of these matrices are disjoint, hence 
$\RPM{A}\begin{smatrix}  0\\I_{n-r}\end{smatrix}=0$ which implies that $A$ has
generic column rank profile.
Similarly, one shows that $\RPM{A}$ can be a sub-matrix of $Q$ only if $A$ has a
generic row rank profile.

\section{Improvements in practice}

In our previous contribution~\cite{DPS:2013}, we identified the ability to
recover the rank profile matrix via the use of the product order search and
of rotations.
Hence we proposed an implementation combining a tile recursive algorithm and an iterative
base case, using these search and permutation strategies.

The analysis of sections~\ref{sec:cond} and~\ref{sec:algo} shows that other
pivoting strategies can be used to compute the rank profile matrix, and preserve
the monotonicity. We present here a new base case algorithm and its
implementation over a finite field that we wrote in the \fflasffpack
library\footnote{FFLAS-FFPACK revision 1193,
  \url{http://linalg.org/projects/fflas-ffpack}, linked against OpenBLAS-v0.2.8.}. 
It is based on a
lexicographic order search and row
and column rotations. Moreover, the schedule of the update operations is that
of a Crout elimination, for it reduces the number of modular reductions, as
shown in~\cite[\S~3.1]{DPS:2014}.
Algorithm~\ref{alg:pluq:iter} summarizes this variant.
\begin{algorithm}[htbp]
  \caption{Crout variant of \pluq with lexicographic search and column rotations}
  \label{alg:pluq:iter}
\begin{algorithmic}[1]
\State $k\leftarrow 1$
\For{$i=1\dots m$}
   \State $A_{i,k..n}\leftarrow A_{i,k..n} - A_{i,1..k-1}\times A_{1..k-1,k..n}$
   \If{$A_{i,k..n} = 0$}
     \State Loop to next iteration
   \EndIf
   \State Let $A_{i,s}$ be the left-most non-zero element of row $i$.
   \State $A_{i+1..m,s}\leftarrow A_{i+1..m,s} - A_{i+1..m,1..k-1}\times A_{1..k-1,s}$
   \State $A_{i+1..m,s}\leftarrow A_{i+1..m,s} / A_{i,s}$
   \State Bring $A_{*,s}$ to $A_{*,k}$ by column rotation
   \State Bring $A_{i,*}$ to $A_{k,*}$ by row rotation
   \State $k\leftarrow k+1$
\EndFor
\end{algorithmic}
\end{algorithm}

\begin{figure}[H]
  \centering
  \includegraphics[width=\columnwidth]{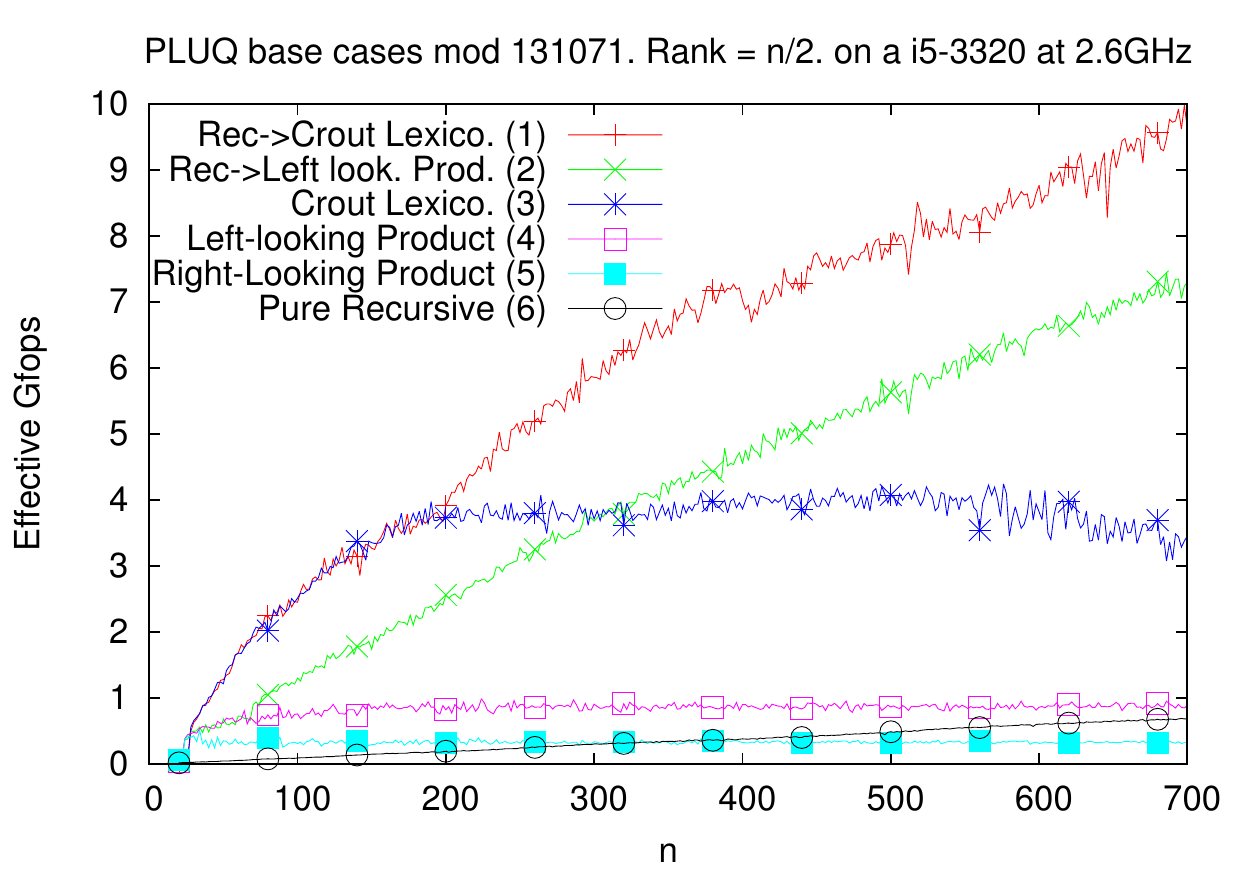}
  \includegraphics[width=\columnwidth]{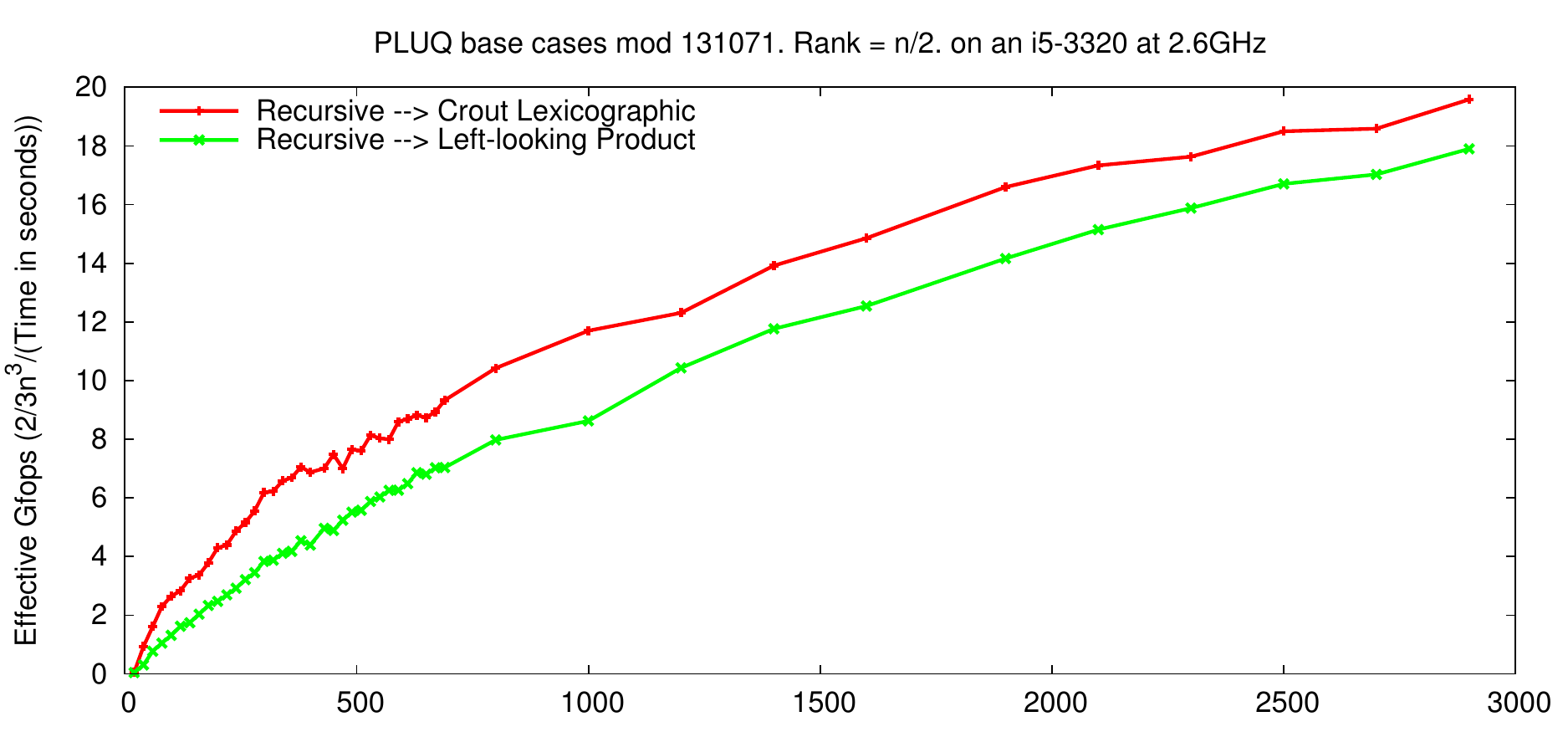}
  \caption{Computation speed of PLUQ decomposition base cases.}
  \label{fig:basecase}
\end{figure}

In the following experiments, we measured the real time of the computation
averaged over 10 instances (100 for $n< 500$) of
$n\times n$ matrices with rank $r=n/2$ for any even integer value of $n$ between 20
and 700. In order to ensure that the row and column rank profiles of these
matrices are uniformly random, we construct them as the product $A=L\RPM{}U$, where
$L$ and $U$ are random non-singular lower and upper triangular matrices and $\RPM{}$ is
an $m\times n$ $r$-sub-permutation matrix whose non-zero elements positions are chosen
uniformly at random.
The effective speed is obtained by dividing an estimate of the arithmetic cost
($2mnr+2/3r^3-r^2(m+n)$) by the computation time. 

Figure~\ref{fig:basecase} shows its computation speed (3),
compared to that of the pure recursive algorithm (6), and to our
previous base case algorithm~\cite{DPS:2013}, using a product order search, and either a
left-looking (4) or a right-looking (5) schedule. At $n=200$, the left-looking
variant (4) improves over the right looking variant (5) by a factor of about
$2.14$ as it performs fewer modular reductions. Then, the Crout variant (3)
again improves variant (4) by a factor of about 3.15. Lastly we also show the speed
of the final implementation, formed by the tile recursive algorithm cascading to either the
Crout base case (1) or the left-looking one (2). The threshold where the
cascading to the base case occurs is experimentally set to its optimum value,
i.e. 200 for variant (1) and 70 for variant (2). This illustrates that the
gain on the base case efficiency leads to a higher threshold, and
 improves the efficiency of the cascade implementation (by an additive
gain of about 2.2 effective Gfops in the range of dimensions considered).

\section{Computing Echelon forms}

Usual algorithms computing an echelon form~\cite{Storjohann:2000:thesis,JPS:2013} use a slab block decomposition (with
row or lexicographic order search), which implies that pivots appear in the
order of the echelon form. The column echelon form is simply obtained as $C=PL$
from the PLUQ decomposition.
Using product order search, this is no longer true, and the order of the columns
in $L$ may not be that of the echelon form. Algorithm~\ref{alg:echelon} shows
how to recover the echelon form in such cases.
\begin{algorithm}[htbp]
  \caption{Echelon form from a PLUQ decomposition}
  \label{alg:echelon}
\begin{algorithmic}[1]
\Require{$P,L,U,Q$, a PLUQ decomp. of $A$ with $\RPM{A}=\Pi_{P,Q}$}
\Ensure{$C$: the column echelon form of $A$}
\State $C\leftarrow PL$
\State $(p_1,..,p_r) =
\text{Sort}(\sigma_P(1),..,\sigma_P(r))$ \label{alg:line:sort}
\For{$i=1..r$}
  \State $\tau = (\sigma_P^{-1}(p_1),..,\sigma_P^{-1}(p_r), r+1,..,m)$
\EndFor
  \State $C\leftarrow C P_\tau$
\end{algorithmic}
\end{algorithm}
Note that both the row and the column echelon forms can thus be computed from
the same PLUQ decomposition.
Lastly, the column echelon form of the $i\times j$ leading
sub-matrix, is computed by removing rows of $PL$ below index $i$ and filtering out the pivots of column index greater than
$j$. The latter is achieved by replacing line~\ref{alg:line:sort} by
 $(p_1,..,p_s) = \text{Sort}( \{\sigma_P(i) : \sigma_Q(i)\leq j\})$.

 {
 \bibliographystyle{abbrvurl}
 \bibliography{pluq}
 }
\end{document}